\documentclass[a4paper,UKenglish,cleveref, autoref, thm-restate]{lipics-v2021}

\title{An Optimal Algorithm for Shortest Paths in Unweighted Disk Graphs\footnote{This paper appeared in {\em Proceedings of the 33rd Annual European Symposium on Algorithms (ESA 2025)}.}}

    \author{Bruce W. Brewer}{Kahlert School of Computing, University of Utah, Salt Lake City, UT 84112, USA}{bruce.brewer@utah.edu}{https://orcid.org/0009-0008-2995-148X}{}

\author{Haitao Wang}{Kahlert School of Computing, University of Utah, Salt Lake City, UT 84112, USA}{haitao.wang@utah.edu}{https://orcid.org/0000-0001-8134-7409}{}

\authorrunning{B. W. Brewer and H. Wang} 

\Copyright{Bruce W. Brewer and Haitao Wang} 



\category{} 

\nolinenumbers 

\graphicspath{{./Figures/}}
\usepackage{cite}
\usepackage{hyperref}
\hypersetup{
    colorlinks=true,
    linkcolor=red,
    filecolor=magenta,
    citecolor = blue,
    urlcolor=cyan,
    linktocpage,
    pagebackref,
    hyperindex = true
    }

\hideLIPIcs
\EventEditors{Anne Benoit, Haim Kaplan, Sebastian Wild, and Grzegorz Herman}
\EventNoEds{4}
\EventLongTitle{33rd Annual European Symposium on Algorithms (ESA 2025)}
\EventShortTitle{ESA 2025}
\EventAcronym{ESA}
\EventYear{2025}
\EventDate{September 15--17, 2025}
\EventLocation{Warsaw, Poland}
\EventLogo{}
\SeriesVolume{351}
\ArticleNo{29}

\usepackage[noend,linesnumbered]{algorithm2e}
\DeclareMathOperator*{\argmin}{arg\,min}
\newcommand{\os}[1]{\overset{(#1)}}

\newcommand{\G}{G(S)} 
\newcommand{\R}{\text{Vor}} 
\newcommand{\E}{\mathcal{E}} 
\newcommand{\DT}{\mathcal{DT}} 
\newcommand{\vd}{\mathcal{VD}}
\newcommand{\N}{N_{\DT}}
\newcommand{\bbR}{\mathbb{R}}

\begin{document}

\maketitle

\begin{abstract}
Given in the plane a set $S$ of $n$ points and a set of disks centered at these points, the disk graph $\G$ induced by these disks has vertex set $S$ and an edge between two vertices if their disks intersect. Note that the disks may have different radii. We consider the problem of computing shortest paths from a source point $s\in S$ to all vertices in $\G$ where the length of a path in $\G$ is defined as the number of edges in the path. The previously best algorithm solves the problem in $O(n\log^2 n)$ time. A lower bound of $\Omega(n\log n)$ is also known for this problem under the algebraic decision tree model. In this paper, we present an $O(n\log n)$ time algorithm, which matches the lower bound and thus is optimal. Another virtue of our algorithm is that it is quite simple.
\end{abstract}

\section{Introduction}
\label{sec:intro}
Let $S$ be a set of $n$ points in the plane. Each point of $S$ is associated with a disk centered at the point. Note that the disks may have different radii. The {\em disk graph} $\G$ induced by these disks has vertex set $S$ and an edge between two vertices if their disks intersect. We consider the single-source-shortest-path (SSSP) problem of computing shortest paths from a source point $s \in S$ to all vertices in $\G$ where the length of a path in $\G$ is defined as the number of edges in the path.

The problem was studied previously. A straightforward solution is to first construct $\G$ explicitly and then run a breadth-first-search (BFS) algorithm, which takes $\Omega(n^2)$ time since $\G$ may have $\Omega(n^2)$ edges in the worst case. As indicated in \cite{ref:KaplanDy20,ref:KaplanTh23}, the problem can be solved by BFS with a dynamic nearest neighbor data
structure for the weighted Euclidean distance; with the data structure in \cite{ref:LiuNe22}, a randomized algorithm of $O(n\log^4 n)$ time can be achieved.
Klost~\cite{ref:KlostAn23} gave a (deterministic) algorithm of $O(n\log^2 n)$ time. On the other hand, it is known that the problem has an $\Omega(n\log n)$ time lower bound under the algebraic decision tree model even if all disks have the same radius~\cite{ref:CabelloSh15}.

In this paper, we present a new algorithm whose runtime is $O(n\log n)$, which matches the $\Omega(n\log n)$ lower bound and thus is optimal. Another virtue of the algorithm is that it is quite simple.
Indeed, the most complex step of our algorithm is constructing a static additively-weighted Voronoi diagram, which can be easily done using Fortune's sweeping algorithm~\cite{ref:FortuneA87}. Aside from that, the algorithm's description in the paper is self-contained.

Concurrently with our work, de Berg and Cabello have also achieved an $O(n \log n)$ time algorithm using a different approach~\cite{ref:deBergAn25}.

\subparagraph{Related work.}
Disk graphs are among the most well-studied topics in computational geometry due to their applications and geometric properties, e.g.,~\cite{ref:AnFa23,ref:ClarkUn90,ref:ChanFa25,ref:LokshtanovSu22,ref:BaumannDy24,ref:EspenantFi23,ref:CabelloMi21,ref:TkachenkoDo25}. The SSSP problem we consider is actually an {\em unweighted case} because the length of each edge of $\G$ is defined to be one. In the {\em weighted} case, the length of each edge is defined to be the Euclidean distance between the two disk centers corresponding to the two vertices of the edge. Kaplan, Katz, Saban, and Sharir~\cite{ref:KaplanTh23} also gave an algorithm for the weighted case, and the algorithm runs in $O(n\log^6 n)$ time. An, Oh, and Xue~\cite{ref:AnSi25} very recently derived an improved algorithm of $O(n\log^4 n)$ time.

The SSSP problem on {\em unit-disk graphs} (in which all disks have the same radius) has also been extensively studied. Cabello and {Jej\v ci\v c}~\cite{ref:CabelloSh15} first solved the unweighted case in $O(n \log n)$ time, matching the $\Omega(n \log n)$ lower bound~\cite{ref:CabelloSh15}. If the points are presorted by $x$- and $y$-coordinates, Chan and Skrepetos~\cite{ref:ChanAl16} showed that it can be solved in $O(n)$ additional time. The weighted case was first solved by Roditty and Segal~\cite{ref:RodittyOn11} in $O(n^{4/3 + \epsilon})$ time for any $\epsilon > 0$. Later, Cabello and {Jej\v ci\v c}~\cite{ref:CabelloSh15} solved it in $O(n^{1 + \epsilon})$ time. The algorithm by Cabello and {Jej\v ci\v c}~\cite{ref:CabelloSh15} is bottlenecked by a dynamic bicromatic closest pair data structure, which has since seen improvements by Kaplan, Mulzer, Roditty, Seiferth, and Sharir~\cite{ref:KaplanDy20} and also by~Liu~\cite{ref:LiuNe22}. Recently, Wang and Xue~\cite{ref:WangNe20} gave a new algorithm of $O(n \log^2 n)$ time. This has been slightly reduced to $O(n \log^2 n / \log \log n)$ by Brewer and Wang~\cite{ref:BrewerAn24}. Furthermore, Brewer and Wang~\cite{ref:BrewerAn24} proved that the problem can be solved using $O(n\log n)$ comparisons under the algebraic decision tree model, matching an $\Omega(n\log n)$ lower bound in the same model~\cite{ref:CabelloSh15}.

Kaplan, Katz, Saban, and Sharir~\cite{ref:KaplanTh23} also studied a ``reversed version'' of the shortest path problem in disk graphs. Chan and Huang~\cite{ref:ChanFa25} very recently improved the algorithm of \cite{ref:KaplanTh23}. Both algorithms used the previous SSSP algorithm for disk graphs as a decision procedure. With our new SSSP algorithm, these algorithms can be slightly improved and also simplified.

In the following, after introducing some notation and concepts in Section~\ref{sec:pre}, we present our algorithm in Section~\ref{sec:algo}.

\section{Preliminaries}
\label{sec:pre}
We follow the notation already defined in Section~\ref{sec:intro}, e.g., $n$, $S$, $s$, $\G$.

We define $S_i$ for $i = 0, 1, 2, \ldots$ to be the set of vertices of $\G$ whose distance from $s$ in $\G$ is equal to $i$. Note that $S_0=\{s\}$. For convenience, we define $S_{\leq i} = \cup_{j \leq i} S_j$.

For any two points $p,q$ in the plane, let $|pq|$ denote the (Euclidean) distance between $p$ and $q$.

To differentiate from other points in the plane, we refer to points of $S$ as {\em sites}. For any site $v\in S$, let $D_v$ denote the input disk centered at $v$ and $r_v$ the radius of $D_v$. Notice that two sites $u,v\in S$ have an edge $(u,v)$ in $\G$ if and only if $|uv|-r_u-r_v\leq 0$. We consider $r_v$ as the {\em weight} of $v$, and for any point $p\in \bbR^2$, we define the {\em (additively) weighted distance} from $p$ to $v$ as $d_v(p)=|vp|-r_v$. Note that if $p$ is outside the disk $D_v$, then $d_v(p)$ is the smallest distance from $p$ to any point of $D_v$.

For any subset $S'\subseteq S$ and a point $p\in \bbR^2$,
we define a \emph{nearest neighbor} of $p$ to be a site in $S'$ whose weighted distance to $p$ is minimal. Formally, we define $\beta_p(S') = \argmin_{v \in S'} d_v(p)$ to be a nearest neighbor of $p$ in $S'$. In cases where $p$ has more than one nearest neighbor in $S'$, we let $\beta_p(S')$ refer to an arbitrary one. We also define $d_{S'}(p)=\min_{v \in S'}d_v(p)$.

With respect to the weights $r_v$ of the sites $v$ of $S$, we will make heavy use of the {\em (additively) weighted Voronoi diagram} of $S$, which is a partition of the plane into Voronoi regions, edges, and vertices~\cite{ref:FortuneA87,ref:SharirIn85}. The Voronoi region $\R(u)$ of a site $u \in S$ is usually defined as the set of points nearer to $u$ than to any other site in $S$. For convenience, we will let $\R(u)$ also include its boundary, defining $\R(u)$ as the set of points at least as near to $u$ than to any other site in $S$. Formally, $\R(u) = \{p \in \bbR^2: \forall\ v \in S \setminus \{u\},\ d_u(p) \leq d_v(p)\}$. This way, the plane is covered by the Voronoi regions. The Voronoi edge between two sites $u, v \in S$ is defined by $\E(u, v) = \{p \in \bbR: \forall\ w \in S \setminus \{u, v\}, d_u(p) = d_v(p) < d_w(p)\}$. Note that $\E(u, v)$ may have multiple connected components, each of which is a straight line segment or a hyperbolic arc~\cite{ref:FortuneA87,ref:SharirIn85}. Voronoi vertices are the points that are equidistant to at least three sites in $S$ such that all other sites in $S$ are farther than these. We use $\vd(S)$ to denote the weighted Voronoi diagram of $S$.

Let $\DT(S)$ denote the dual graph of $\vd(S)$. Specifically, $S$ is the vertex set of $\DT(S)$, and the edge set of $\DT(S)$ consists of all pairs of sites $u, v \in S$ whose Voronoi regions share a Voronoi edge, i.e., $\E(u, v) \neq \emptyset$. Note that if the sites of $S$ all have the same weight (i.e., all disks have the same radius), then $\DT(S)$ is the Delaunay triangulation of $S$. We define the neighborhood $\N(u)$ of a site $u \in S$ to be the set of sites that share an edge with $u$ in $\DT(S)$. By extension, we define the neighborhood $\N(S')$ of a set $S' \subseteq S$ by $\N(S') = \bigcup_{u \in S'} \N(u)$.

For ease of exposition, we make the general position assumption that no three sites of $S$ lie on the same line. The assumption can be lifted without affecting the results, e.g., by standard perturbation techniques~\cite{ref:EdelsbrunnerSi90,ref:YapA90}.

\section{The algorithm}
\label{sec:algo}
In this section, we present our algorithm to compute the sets $S_i$ for $i=0,1,2,\ldots$. It is easy to extend our algorithm to also compute predecessor information for the shortest paths.

We follow the same algorithmic scheme of Cabello and {Jej\v ci\v c}~\cite{ref:CabelloSh15} for unit-disk graphs and manage to extend their algorithm to our general disk graph case. In what follows, we first describe the algorithm, then discuss the implementation and time analysis, and finally prove its correctness.

We assume that no disk fully contains another disk in the input. With this assumption, every site of $S$ defines a non-empty Voronoi region in $\vd(S)$~\cite{ref:SharirIn85}. We will later explain that the algorithm can be slightly modified to accommodate the situation in which the assumption does not hold.

\subsection{Algorithm description}
We begin with constructing the weighted Voronoi diagram $\vd(S)$. Then the algorithm runs in iterations, and each $i$-th iteration will compute the set $S_i$. Initially, we have $S_0=\{s\}$. In the $i$-th iteration for $i\geq 1$, we compute $S_i$ as follows, assuming that $S_{\leq i-1}$ has already been computed. Let $Q=\N(S_{i-1})\setminus S_{\leq i-1}$, i.e., the neighbors of $S_{i - 1}$ in $\DT(S)$ that are not in $S_{i - 1}$. We remove a site $v\in Q$ from $Q$ and determine whether $d_{S_{i-1}}(v)\leq r_v$. If it is, then we add $v$ to $S_i$ and update $Q=Q\cup (\N(v)\setminus S_{\leq i})$, i.e., we add the neighbors of $v$ in $\DT(S)$ that are not in $Q\cup S_{\leq i}$ to $Q$. We repeat this until $Q=\emptyset$, at which point the $i$-th iteration is complete. After the $i$-th iteration, we check if $S_i = \emptyset$. If so, then we are done, and the algorithm stops. Otherwise, we continue to the $i+1$-th iteration. Algorithm~\ref{algo:sssp} gives the pseudocode.

\begin{algorithm}
    \caption{SSSP algorithm} \label{algo:sssp}
    \KwIn{A set of point sites $S$, radius $r_v$ for each site $v\in S$, source site $s \in S$.}
    \KwOut{The sets $S_i$ for $i = 0, 1, \ldots$}
    $i \gets 0$\;
    $S_0 \gets \{s\}$\;
    Build $\vd(S)$ and $\DT(S)$\;
    \While{$S_i \neq \emptyset$}{
        $i \gets i + 1$\;
        $S_i \gets \emptyset$\;
        $Q \gets \N(S_{i - 1}) \setminus S_{\leq i-1}$ \label{line:Q1}\;
        \While{$Q \neq \emptyset$ \label{line:Q_while}}{
            $v \gets pop(Q)$\;
            \If{$d_{S_{i - 1}}(v) \leq r_v$ \label{line:check}}{
                $S_i \gets S_i \cup \{v\}$\; \label{line:add}
                $Q \gets Q \cup (\N(v) \setminus S_{\leq i})$ \label{line:Q2}\;
            }
        }
    }
    \Return{$S_i$ for $i = 0, 1, \ldots$\;}
\end{algorithm}

\subsection{Algorithm implementation and time analysis}
We now discuss how to implement the algorithm efficiently.

First of all, computing $\vd(S)$ and thus $\DT(S)$ can be done in $O(n\log n)$ time~\cite{ref:FortuneA87}. Note that the combinatorial sizes of both $\vd(S)$ and $\DT(S)$ are $O(n)$~\cite{ref:FortuneA87,ref:SharirIn85}.

Next, notice that the number of vertices added to $Q$ on Line~\ref{line:Q1} in iteration $i$ is at most the sum of the degrees of the vertices of $S_{i - 1}$ in the graph $\DT(S)$. Because the sets $S_i$ are disjoint and $\DT(S)$ has $O(n)$ edges, this means that the number of vertices added to $Q$ by Line~\ref{line:Q1} throughout the whole algorithm is $O(n)$. By the same reasoning, we conclude that the number of vertices added to $Q$ by Line~\ref{line:Q2} throughout the whole algorithm is also $O(n)$. We can also know that the while loop on Line~\ref{line:Q_while} has at most $O(n)$ iterations throughout the whole algorithm because a vertex is removed from $Q$ in each iteration.

It remains to determine the time complexity for the operation of checking whether $d_{S_{i-1}}(v)\leq r_v$ on Line~\ref{line:check}. To this end, it suffices to find $\beta_v(S_{i-1})$, i.e., the nearest neighbor of $v$ in $S_{i-1}$. To find $\beta_v(S_{i-1})$, in the beginning of the $i$-th iteration, we compute the weighted Voronoi diagram for $S_{i-1}$, denoted by $\vd(S_{i-1})$, which takes $O(|S_{i-1}|\log |S_{i-1}|)$ time~\cite{ref:FortuneA87}, and then construct a point location data structure for $\vd(S_{i-1})$ in $O(|S_{i-1}|)$ time~\cite{ref:KirkpatrickOp83,ref:EdelsbrunnerOp86}. After that, $\beta_v(S_{i-1})$ can be found in $O(\log n)$ time by a point location query with $v$ on $\vd(S_{i-1})$. As such, the total time of Line~\ref{line:check} in the whole algorithm is $O(n\log n)$. As $\sum_i|S_{i-1}|=n$, the overall time for constructing $\vd(S_{i-1})$ in the whole algorithm is $O(n\log n)$.

We thus conclude that the time complexity of the whole algorithm is $O(n\log n)$.

\subsection{Algorithm correctness}
We start with proving the following two observations that our algorithm's correctness relies on.

\begin{observation}\label{obser:weightdis}
    For any two sites $u,v\in S$, they are connected by an edge in $\G$ if and only if $d_v(u)\leq r_u$ or $d_u(v)\leq r_v$.
\end{observation}
\begin{proof}
    Recall that $u$ and $v$ have an edge if and only if $|uv|-r_u-r_v\leq 0$.
    Notice that $|uv|-r_u-r_v\leq 0$ holds if and only if $d_v(u)\leq r_u$ or $d_u(v)\leq r_v$.
\end{proof}

\begin{observation}\label{obser:nearneighbor}
    For any site $v\in S\setminus S_{\leq i-1}$, $v$ is in $S_i$ if and only if $d_{S_{i-1}}(v)\leq r_v$.
\end{observation}
\begin{proof}
    By definition, $S_{i-1}$ has a site $u$ with $d_u(v)=d_{S_{i-1}}(v)$. If $d_{S_{i-1}}(v)\leq r_v$, then $d_u(v)\leq r_v$. By Observation~\ref{obser:weightdis}, $\G$ has an edge connecting $u$ and $v$. As $u\in S_{i-1}$ and $v\in S\setminus S_{\leq i-1}$, we obtain that $v\in S_i$.

    If $v\in S_i$, then there is an edge in $\G$ connecting $v$ to a site $u\in S_{i-1}$. By Observation~\ref{obser:weightdis}, $d_u(v)\leq r_v$. It then follows that $d_{S_{i-1}}(v)\leq d_u(v)\leq r_v$.
\end{proof}

We now prove that the algorithm is correct. Let $S_i'$ be the set $S_i$ computed by our algorithm and let $S_i$ be the ``correct'' set of vertices at distance $i$ from $s$ in $\G$. Then, our goal is to prove that $S_i' = S_i$ for all $i$. We do so by induction on $i$. Our base case is that $S_0' = S_0$, which is obviously true because they are both $\{s\}$. Our inductive hypothesis is that $S_j' = S_j$ for all $j < i$. To prove that $S_i' = S_i$, we will argue both $S_i' \subseteq S_i$ and $S_i' \supseteq S_i$.

First of all, since $S_j'=S_j$ for all $j<i$, our algorithm guarantees that every site of $S_i'$ is at distance $i$ from $s$ in $\G$ and thus $S_i'\subseteq S_i$. To see this, a site $v$ is added to $S_i'$ in Line~\ref{line:add} because $d_{S_{i-1}}(v)\leq r_v$. By Observation~\ref{obser:nearneighbor}, $v\in S_i$.

We next prove that $S_i' \supseteq S_i$. Consider a site $v\in S_i$. Our goal is to show that $v$ is also in $S_i'$.
We will prove in \cref{lem:path} that there exists a vertex $u \in S_{i - 1}$ and a $u$-$v$ path $\pi$ in $\DT(S)$ such that all internal vertices of $\pi$ are in $S_i$ (note that {\em internal vertices} refer to vertices of $\pi$ excluding $u$ and $v$). Let the vertices of $\pi$ be $u=w_0, w_1, \ldots, w_k=v$. We argue that all $w_j$ with $1 \leq j \leq k$ are added to $S_i'$. To this end, because $w_j \in S_i$ for all $1 \leq j \leq k$, it is sufficient to argue that $w_j$ is added to $Q$ because  when $w_j$ is popped from $Q$, it will pass the check on Line~\ref{line:check} by \cref{obser:nearneighbor} and then be added to $S_i'$. Indeed, $w_1$ is added to $Q$ on Line~\ref{line:Q1} because $w_1 \in \N(u)$.
Since $w_2\in \N(w_1)$, $w_2$ will be added to $Q$ on Line~\ref{line:Q2} no later than right after $w_1$ is added to $S'_i$ on Line~\ref{line:add}. Following the same argument, $w_3,w_4,\ldots,w_k=v$ will all be added to $Q$ and thus added to $S_i'$ as well. This proves that $v\in S_i'$.
Therefore, $S_i' \supseteq S_i$.

In summary, we have proved $S_i' = S_i$ and thus the correctness of the algorithm.

\begin{lemma} \label{lem:path}
    For any $i > 0$ and any site $v \in S_i$, there exist a site $u \in S_{i - 1}$ and a $u$-$v$ path in $\DT(S)$ such that all internal vertices of the path are in $S_i$.
\end{lemma}
\begin{proof}
    We extend the proof of \cite[Lemma 1]{ref:CabelloSh15} for the unit-disk graph case to our general disk graph problem.

   Let $u =\beta_v(S_{i-1})$, i.e., $u$ is a nearest neighbor of $v$ in $S_{i-1}$.
    Consider the line segment $\overline{uv}$ and the sites whose Voronoi regions in $\vd(S)$ intersect $\overline{uv}$. Let $w_0, w_1, \ldots, w_k$ be these sites in the order that their Voronoi regions intersect $\overline{uv}$ with $w_0 = u$ and $w_k = v$. Note that repetitions are possible, so it could be that $w_j = w_{j'}$ for $j \neq j'$. See \cref{fig:path} for an example. For ease of discussion, we assume that $\overline{uv}$ does not contain any Voronoi vertices of $\vd(S)$ (otherwise, we could replace $v$ by a point arbitrarily close to $v$ and then follow the same argument). This implies that, for every $0\leq j\leq k-1$, $\overline{uv}$ crosses the Voronoi edge between $w_j$ and $w_{j + 1}$. This further implies that $w_j$ and $w_{j + 1}$ are adjacent in $\DT(S)$, so $\pi = w_0, w_1, \ldots, w_k$ is a $u$-$v$ path in $\DT(S)$.

    \begin{figure}[t]
        \centering
        \includegraphics[height = 5cm]{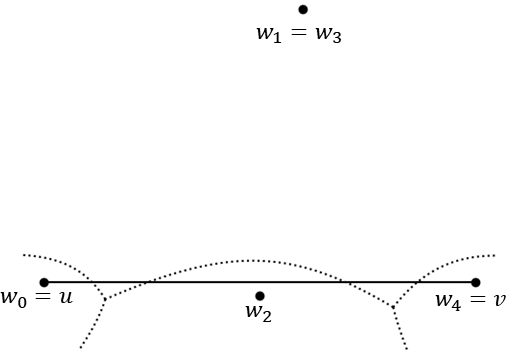}
        \caption{The segment $\overline{uv}$ passes through $\R(u)$, $\R(w_1)$, $\R(w_2)$, $\R(w_3)$, and $\R(v)$, where $w_1 = w_3$.}
        \label{fig:path}
    \end{figure}

    Next, we need to show that all internal vertices $w_j$ ($1 \leq j \leq k - 1$) of $\pi$ are in $S_i$. We will first prove that there is an edge $(u,w_j)$ in $G(S)$ connecting $u$ and $w_j$, meaning that $w_j \in S_{i - 2} \cup S_{i - 1} \cup S_i$ since $u\in S_{i-1}$. Then, we will show that $w_j \notin S_{i - 2} \cup S_{i - 1}$, so $w_j \in S_i$ must hold. This will complete the proof of the lemma.

    \begin{figure}[t]
        \centering
        \includegraphics[height = 2.5cm]{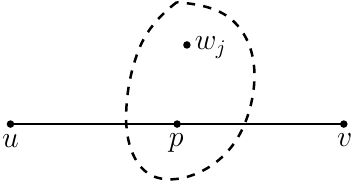}
        \caption{The region bounded by the dashed curve is $\R(w_j)$.}
        \label{fig:segment}
    \end{figure}

\subparagraph{Proving $\boldsymbol{G(S)}$ has an edge $\boldsymbol{(u,w_j)}$.}
We first argue that $G(S)$ has an edge $(u,v)$. Indeed, since $v\in S_i$,
there is a vertex $u'\in S_{i-1}$ connecting to $v$ by an edge in $\G$. By Observation~\ref{obser:weightdis}, $d_{u'}(v)\leq r_v$. Since $u=\beta_v(S_{i-1})$, we have $d_u(v)\leq d_{u'}(v)\leq r_v$. By Observation~\ref{obser:weightdis}, $\G$ has an edge $(u,v)$.

Consider any $w_j$ with $1\leq j\leq k-1$. We now argue that $G(S)$ has an edge $(u,w_j)$.
    Recall that $\overline{uv}$ intersects $\R(w_j)$. Let $p$ be a point of $\overline{uv} \cap \R(w_j)$; see Figure~\ref{fig:segment}. We can derive the following:
    \[d_{w_j}(u) = |w_ju|-r_{w_j} \os 1 \leq |w_j p| + |p u| - r_{w_j} \os 2 \leq |v p| + |p u| - r_v \os 3 = |vu|-r_v=  d_v(u) \os 4 \leq r_u,\]
    where (1) is due to the triangle inequality, (2) is due to $p \in \R(w_j)$ (therefore, $|w_j p| - r_{w_j} = d_{w_j}(p) \leq d_v(p) = |v p| - r_v$), (3) is due to $p \in \overline{u v}$, and (4) is due to \cref{obser:weightdis} and the fact that $\G$ has an edge $(u,v)$. Consequently, by \cref{obser:weightdis}, $G(S)$ has an edge $(u,w_j)$.

\subparagraph{Proving $\boldsymbol{w_j \notin S_{i - 2} \cup S_{i - 1}}$.}
    To prove that $w_j \notin S_{i - 2} \cup S_{i - 1}$, $1\leq j\leq k-1$, as before, we pick a point $p \in \overline{u v} \cap \R(w_j)$ and derive the following:
    \[d_{w_j}(v) \os 1 < |w_j p| + |p v| - r_{w_j} \os 2 \leq |u p| + |p v| - r_u \os 3 = |uv| - r_u= d_u(v) \os 4 \leq r_v,\]
    where (1) is due to the triangle inequality and our general position assumption (because $p \in \overline{u v}$, equality occurs only if $u$, $v$, and $w_j$ are collinear, violating our general position assumption), (2) is due to $p \in \R(w_j)$ (therefore, $|w_j p| - r_{w_j} = d_{w_j}(p) \leq d_u(p) = |u p| - r_u$), (3) is due to $p \in \overline{u v}$, and (4) is due to \cref{obser:weightdis} and the fact that $\G$ has an edge $(u,v)$, as proven above.

    We find that $w_j \notin S_{i - 1}$ since otherwise $d_{w_j}(v) < d_u(v)$ proved above would contradict with $u = \beta_v(S_{i - 1})$. Also, since $G(S)$ has an edge $(w_j,v)$ and $v\in S_i$, $w_j$ cannot be in $S_{i-2}$ (since otherwise $v$ would be in $S_{\leq i-1}$, contradicting with that $v\in S_i$). This proves that $w_j \notin S_{i - 2} \cup S_{i - 1}$.
\end{proof}

\subparagraph{The containment case.}
The above discusses the case where no disk fully contains another disk in the input. Suppose that a disk is contained in another disk. Then, the Voronoi region of the center of the smaller disk in $\vd(S)$ is empty~\cite{ref:SharirIn85}, and our algorithm will miss the disk. To fix the issue, we can slightly revise our algorithm as follows. Let $v$ be a site of $S$ whose Voronoi region in $\vd(S)$ is empty. Let $u$ be the site of $S$ whose Voronoi region contains $v$ (note that the disk of $v$ must be contained in the disk of $u$). Then, we add $v$ to $\N(u)$ and also add $u$ to $\N(v)$. After doing this for all such $v$, we apply exactly the same algorithm as before. One can verify that the algorithm works even if the source site $s$ defines an empty Voronoi region in $\vd(S)$. The total time of the algorithm is still $O(n\log n)$.
\medskip

The following theorem summarizes our result.

\begin{theorem}
    Given a set of $n$ disks in the plane and a source disk, shortest paths from the source to all other vertices in the disk graph induced by all disks can be computed in $O(n\log n)$ time.
\end{theorem}

\subparagraph{Extension to the $\boldsymbol{L_{\infty}}$ (resp., $\boldsymbol{L_1}$) metric.}
In the $L_{\infty}$ metric, each disk becomes a square centered at a site of $S$. For this case, an $O(n\log n)$-time algorithm was given by Klost~\cite{ref:KlostAn23} to solve the SSSP problem, which uses the involved techniques in \cite{ref:BaumannDy24}. We can solve this problem in $O(n\log n)$ time in a much simpler way by making the following changes to Algorithm~\ref{algo:sssp}. First, use the $L_{\infty}$ metric to measure the distance $|pq|$ of any two points $p,q$ in the plane. Second, use additively-weighted Voronoi diagrams in the $L_{\infty}$ metric instead. Note that constructing an additively-weighted Voronoi diagram for a set of $n$ points in the $L_{\infty}$ metric is equivalent to constructing a Voronoi diagram for a set of $n$ squares in the $L_{\infty}$ metric. This latter problem can be solved in $O(n\log n)$ time, e.g., by the algorithm of Papadopoulou and Lee~\cite{ref:PapadopoulouTh01}. As such, Algorithm~\ref{algo:sssp} can be modified to solve the $L_{\infty}$ metric case in $O(n\log n)$ time. The algorithm is much simpler than the one in \cite{ref:KlostAn23}. The $L_1$ case can be treated analogously.

\bibliography{references}

\appendix

\end{document}